\documentclass[runningheads,a4paper]{llncs}

\usepackage{enumerate,amsmath,amssymb,graphicx}
\usepackage{fancyvrb}
\usepackage{epsfig}
\usepackage{tikz,xspace,multicol}
\usetikzlibrary{positioning, fit, calc}
\usepackage[colorlinks=true]{hyperref}%
\usepackage{cleveref}
\usepackage{enumitem}
\usepackage{algpseudocode,algorithm}

\usepackage{thmtools,thm-restate}

\usepackage{amssymb}
\usepackage{tikz,xspace,multicol,pgfplots}
\pgfplotsset{compat=newest}

\usetikzlibrary{positioning,shapes,arrows,automata}
\tikzset{ state/.style={draw,ellipse,initial text=} }

\tikzset{
  loop above right/.style={above right, out= 60, in= 30, loop},
  loop above left/.style ={above left,  out=150, in=120, loop},
  loop below right/.style={below right, out=330, in=300, loop},
  loop below left/.style ={below left,  out=240, in=210, loop}
}

\usetikzlibrary{positioning,intersections,shapes,arrows,automata}
\tikzset{ state/.style={draw,ellipse,initial text=} }

\tikzset{
  loop above right/.style={above right, out= 60, in= 30, loop},
  loop above left/.style ={above left,  out=150, in=120, loop},
  loop below right/.style={below right, out=330, in=300, loop},
  loop below left/.style ={below left,  out=240, in=210, loop}
}

\renewcommand{\Game}{\mathcal{G}}
\newcommand{\psmas}{PSMAS\xspace}

\newcommand{\rpatl}	{{\sf PATL-SR}\xspace}

\newcommand{\ie}    	{i.e., }

\newcommand{\wrt}   	{w.r.t.\ }
\newcommand{\true}   	{{\mathtt{true}}}

\newcommand{\prob}   	{{\sf {Prob}}} 

\newcommand{\dist}   	{{\sf{Dist}}}

\newcommand{\rew}   	{{\sf{rew}}}

\def\F{\Diamond} 
\def\G{\Box} 
\def\P{\mathbf{P}}
\def\R{\mathbf{R}}
\def\D{\mathbf{D}}
\def\U{\mathbf{U}}

\def\X{\bigcirc}

\newcommand{\HIST}[2]{\mathsf{Hist}_{#2}(#1)}

\newcommand{\ACT}{\mathsf{Act}}
\newcommand{\AG}{\mathsf{Ag}}
\newcommand{\ASP}{\mathsf{Ap}}

\newcommand{\payoff}{\wp}
\newcommand{\BR}{\mathbf{BR}}

\newcommand{\BCR} {{\mathsf{bCR} }}

\newcommand{\init} {\textit{\small init}}
\newcommand{\crash} {\textit{\small crash}}

\newcommand{\pass} {\textit{\small pass}}

\newcommand{\FS}{\rightarrow}
\newcommand{\CAL}[1]{\mathcal{#1}}
\newcommand{\EEE}{\mathbb{E}}
\newcommand{\DDD}{\mathbb{D}}
\newcommand{\FFF}{\CAL{F}}

\newcommand{\MMM}{\CAL{M}}

\newcommand{\NNN}{\mathbb{N}}
\newcommand{\PPP}{\mathbb{P}}
\newcommand{\RRR}{\mathbb{R}}

\newcommand{\TRANS}[1]{\stackrel{#1}{\longrightarrow}}
\newenvironment{GRAMMAR}{\[\begin{array}{lcl}}{\end{array}\]}
\newcommand{\VERTICAL}{\  \mid\hspace{-3.0pt}\mid \ }

\urldef{\mailsa}\path|Chunyan.Mu@abdn.ac.uk, M.Najib@hw.ac.uk|
\newcommand{\keywords}[1]{\par\addvspace\baselineskip
\noindent\keywordname\enspace\ignorespaces#1}

\newcommand{\jact}{\vec{a}}

\newcommand{\strElm}{\sigma}

\newcommand{\strpElm}{\vec{\sigma}}

\newcommand{\aval}{\Upsilon}

\begin{document}

\title{Counterfactual Reasoning for 
\\Causal Responsibility Attribution in\\
Probabilistic Multi-Agent Systems}
\titlerunning{Counterfactual Reasoning for CR Attribution\\ in Probabilistic MASs}
 \author{Chunyan Mu$^1$
\and
Muhammad Najib$^2$}
\authorrunning{Mu \and Najib}
 \institute{$^1$University of Aberdeen \quad
$^2$Heriot-Watt University \\
\mailsa}

\maketitle

\begin{abstract}
    Responsibility allocation---determining the extent to which agents are accountable for outcomes---is a fundamental challenge in the design and analysis of multi-agent systems. In this work, we model such systems as concurrent stochastic multi-player games and introduce a notion of retrospective (backward) counterfactual responsibility, which quantifies an agent's accountability for outcomes resulting from a given strategy profile. To allocate responsibility among agents, we utilise the Shapley value and formally show that this method satisfies key desirable properties, including fairness and consistency.
    Building on this foundation, we propose a formal framework that supports both verification and strategic reasoning in responsibility-aware multi-agent systems. Furthermore, by adopting Nash equilibrium as the solution concept, we demonstrate how to compute stable strategy profiles in which agents trade off responsibility against expected reward.
\end{abstract}


\section{Introduction}
\label{sec:intro}

The growing emphasis on \textit{trustworthy AI}~\cite{chatila2021trustworthy,li2023trustworthy} necessitates a paradigm shift, particularly regarding \textit{responsibility} and \textit{accountability}~\cite{dignum2020responsibility,kaur2022trustworthy,belle2024trustworthiness}. 
To enable responsible behaviour in agents, they must be endowed with an awareness of responsibility. 
This motivated the notion of \textit{responsibility-aware agents}~\cite{yazdanpanah2023reasoning,kobayashi2023formal}, which are assumed to consider not only their own rewards/payoffs but also the broader implications of their actions. In MAS setting, where agents interact and may collaborate, the problem of responsibility attribution becomes increasingly important. When multiple agents jointly bring about an outcome, how can we fairly attribute responsibility to each agent, taking into account their respective capability? To illustrate this challenge, consider the following example.
\begin{example}
\label{eg:case}
Consider two autonomous vehicles, $ A_1 $ and $ A_2 $, approaching an uncontrolled junction under snowy conditions. $ A_1 $ is travelling from east to west, $ A_2 $ from south to north. Each vehicle can either brake or not brake. If neither brakes, they will crash into each other. If both successfully brake (and stop), they will avoid a collision. However, due to the slippery road conditions, there is a $ 0.2 $ chance that $ A_1 $ will not be able to stop, and a $ 0.6 $ chance that $ A_2 $ will not be able to stop. Each vehicle aims to minimise its travel time to reach its destination. However, a crash would result in a significant cost.
\end{example} 
Suppose the outcome is a collision resulting from the agents' chosen actions (or strategies), \textit{which vehicle should bear greater responsibility for the collision?} Naturally, this depends on the strategies adopted during the interaction. However, responsibility may also hinge on the agents' \textit{capabilities}~\cite{gerstenberg2011blame}.
To illustrate, suppose that $ A_2 $ can guarantee stopping with certainty. If both agents adopt the same mixed strategy that includes a non-zero probability of not braking, and a collision occurs, it seems reasonable to assign greater responsibility to $ A_2 $, as it had the capability to avoid the collision entirely. 
Moreover, there is an inherent trade-off between minimising travel time and avoiding collisions. This creates a strategic decision-making problem in which agents must weigh their utility (e.g., travel time) against potential responsibility penalties for a collision. 
Given these considerations, some important questions arise: \textit{Given the strategy profiles and capabilities, how should responsibility be allocated?}, \textit{If the agents are aware of responsibility attribution, how should they choose their strategies?}

In this paper, we aim to such questions. First, we study responsibility attribution through the lens of \textit{counterfactual reasoning}, which evaluates an agent's contribution by asking how the outcome would change if the agent had acted differently. 
We follow well-established accounts of causation in philosophy and AI, where counterfactuals capture ``\textit{what would have happened otherwise}'' \cite{ChocklerH//:04,lewis2013counterfactuals,pearl2009causality}. 
We then examine strategic reasoning for responsibility-aware agents---agents that weigh not only their expected rewards but also the potential responsibility they may incur. This perspective provides a framework for addressing the earlier question of (strategic) \textit{forward reasoning}: how such agents should choose their strategies.

\paragraph{Contribution}
Our main \textit{contribution} is an approach for responsibility attribution in probabilistic multi-agent systems. Firstly, we introduce a formal framework based on the \textit{Shapley value}~\cite{Shapley//:52} and show that the resulting allocations satisfy desirable properties such as \textit{fairness} (responsibility reflects contribution) and \textit{consistency} (an agent's share does not decrease when its influence increas). Although we do not claim this to be a definitive account of responsibility, the framework is interpretable, and based on a concept used widely across disciplines. Our focus is not on formalising \textit{moral} responsibility, an extremely complex area with no clear consensus, nevertheless, we believe that our approach can contribute toward such a formalisation.

Secondly, we study strategic reasoning in the setting where agents employ \textit{randomised memoryless strategy} (aka \textit{behavioural}~\cite{cristau2010we}), i.e., functions that map from state to a probability distribution over actions. Following the tradition of formal method in MAS, we extend the logic PATL with \textit{cumulative reward} (similar to rPATL~\cite{KwiatkowskaGPS//:19}) and introduce \textit{degree of responsibility operators}. We show that model checking this logic remains in PSPACE, thus no harder than rPATL \cite{KwiatkowskaNPS//:21}. Finally, we present a method for computing Nash equilibrium strategy profiles where agents' utility functions incorporate both reward and responsibility allocation. We show that when agents' temporal objectives are specified as either \textit{reachability} or \textit{safety} conditions, these equilibria can also be computed in PSPACE.

\paragraph{Related work}
Research on responsibility and attribution spans a large body of work; here, we highlight some relevant studies and key distinctions.
Responsibility attribution is closely tied to causality and counterfactual reasoning. Foundational work by Halpern and Pearl~\cite{halpern2005causes} introduced structural causal models, which underpin formal definitions of causality and counterfactual dependence. Building on this, \cite{ChocklerH//:04} proposed a quantitative notion of responsibility, which later extended to a notion of \textit{blame attribution}~\cite{HalpernK//:18}. While conceptually related to our approach, their framework relies on directed acyclic graphs and is primarily suited for reasoning about beliefs, intentions, and epistemic states, whereas our work focuses on temporally extended properties in strategic settings. 

The Shapley value has also been applied for responsibility allocation. \cite{FriedenbergH//:19} extends the structural-model approach \cite{HalpernK//:18} by incorporating the Shapley value to distribute blame, and demonstrating that the allocation scheme satisfies desirable properties. As previously discussed, their model differs from ours and is suited for different purposes. Similarly, \cite{YazdanpanahDJAL//:19} apply Shapley-based responsibility allocation in concurrent epistemic games, enabling epistemic modeling but under deterministic assumptions, in contrast to our probabilistic setting.
\cite{Baier0M//:21b} studies backward and forward responsibility in \textit{extensive-form games}, while \cite{Baier0M//:21a} employs parametric Markov chains (pMCs), which are closely related to our model. Memoryless strategy profiles naturally correspond to parameters assignements in a pMC, although the converse does not always hold\footnote{Some parameter assignments cannot be expressed as individual strategies; see \cite[Remark 8]{stan2024concurrent}. We address this by imposing additional constraints in our parametric model (\Cref{sec:resp-aware}).}. 
Furthermore, \cite{Baier0M//:21a} primarily focuses on attributing changes to \textit{parameters}, with its main concern being how variations in temporal satisfaction degrees should be attributed to changes in parameter assignments. In contrast, our work assigns responsibility to individual players.

Closer to our work is \cite{mu2025responsibility}, which uses concurrent stochastic games to reason about \textit{causal active} and \textit{causal passive} responsibilities, building on \cite{ParkerGL//:23}. Their approach considers how strategies contribute to or prevent outcomes but does not employ counterfactual reasoning nor guarantee fair attribution. From temporal-logic perspective, \cite{de2025responsibility} uses LTL$_f$ to represent outcomes---similar to ours---but focuses on a single agent from a first-person perspective, whereas our approach adopts a third-person, multi-agent view. Other related works include Naumov and Tao's logical frameworks for blameworthiness and ``seeing-to-it'' responsibility~\cite{naumov2019blameworthiness,NaumovT//:21}, as well as STIT‑based accounts~\cite{ciuni2009attributing,de2010logic,baltag2021causal,abarca2022stit}. These differ from ours in both their underlying models and their conceptions of responsibility. In particular, STIT approaches define responsibility via an agent’s modal ability to “see to it that’’ an outcome occurs, whereas our framework evaluates responsibility through counterfactual dependence---what would happen were an agent to act differently---and seeks a fair quantitative attribution across agents.

\section{Preliminaries}
\label{sec:model}

\subsection{Concurrent stochastic game}
\begin{definition}
\label{def:model}
A \textit{concurrent stochastic multi-player game} (CSG) is a tuple $ \Game = (\AG, S, s^{0}, (\ACT_i)_{i \in \AG}, \delta, \ASP, L )$ where:
\begin{itemize}
	\item $\AG=\{1, \dots, n\}$ is a finite set of \emph{agents};
	\item $S$ is a finite non-empty set of \emph{states};
	\item $s^0 \in S$ is the \emph{initial state};
	\item $ \ACT_i $ is a finite set of actions for $ i $. With each agent $ i $ and state $ s \in S $, we associate a non-empty set $ \ACT_i(s) $ of \textit{available} actions that $ i $ can perform in $ s $. Write $ {\ACT}^{\AG} = \ACT_1 \times \cdots \times \ACT_n $.
	\item $ \Delta: S \times {\ACT}^{\AG} \to \dist(S) $ is a probabilistic transition function;
	\item $\ASP$ is a finite set of \emph{atomic propositions}; 
	
	\item $L: S \FS 2^{\ASP}$ is the labelling function mapping each state to a set of atomic propositions drawn from $\ASP$.
\end{itemize}
\end{definition}
Following~\cite{KwiatkowskaNPS//:21}, we augment CSGs with \textit{reward structures} of the form $r=(r_s, r_a)$, where $r_s: S \to \RRR$ is the state reward and $r_a: \ACT^{\AG} \to \RRR$ is the action reward, and consider \textit{cumulative rewards}, that is the sum of payoffs accumulated during the run until a specific point. 

\begin{figure}[t]
 \begin{center}

\scalebox{0.8}{
	\begin{tikzpicture}[->,>=stealth',shorten >=1pt,auto,node distance=2cm,
		thick,
		main node/.style={circle,draw,font=\Large\bfseries},
		action node/.style={rectangle,draw,font=\Large\bfseries}
		]
		
		\node[main node] (0) {$s_0$};
		\node[action node] (2) [below left=1cm of 0] {};
		\node[action node] (3) [below right=1cm of 0] {};

		\node[main node] (5) [below=1.5cm of 2] {$s_1$};
		\node[action node] (1) [left=2cm of 5] {};
		\node[main node] (6) [below=1.5cm of 3] {$s_2$};
		\node[action node] (4) [right=2cm of 6] {};
%
		\path(0) edge[bend right=40] node [left] {$\bar{b}_1, \bar{b}_2$} (1);
		\path(0) edge node [left] {$\bar{b}_1, b_2$} (2);
		\path(0) edge node [right] {${b}_1, \bar{b}_2$} (3);
		\path(0) edge[bend left=40] node [right] {${b}_1, {b}_2$} (4);
		\path(1) edge node [above] {$ 1 $} (5);
		\path(2) edge node [left] {$ 0.6 $} (5);
		\path(2) edge [bend right=10] node [left] {$ 0.4 $} (6);
		\path(3) edge [bend left=10] node [right] {$ 0.2 $} (5);
		\path(3) edge node [right] {$ 0.8 $} (6);
		\path(4) edge [bend left=50] node [below] {$ 0.12 $} (5);
		\path(4) edge node [above] {$ 0.88 $} (6);
        \path(5) edge [loop below] node {$\ast, \ast$} ();
        \path(6) edge [loop below] node {$\ast, \ast$} ();
	\end{tikzpicture}
}

 \caption{CSG model from the running example.}
 \label{fig:eg-model}
 \end{center}
 \vspace{-0.5cm}
 \end{figure}
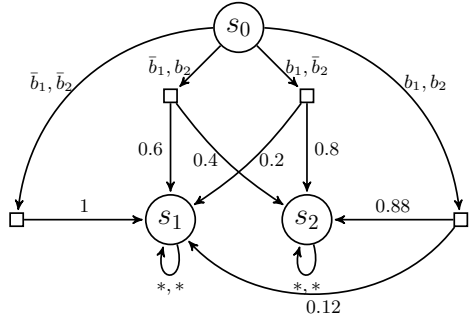
\begin{example}
\label{eg:model}
From the running example, we construct a model in \Cref{fig:eg-model}. With $ \AG = \{1,2\} $, representing $ A_1 $ and $ A_2 $. $ \ACT_i = \{b_i, \bar{b}_i\} $ representing agent $ i \in \AG $ choosing to brake and not brake. $ s_0 $ is the initial state, and $ L(s_1) = \{\crash\}, L(s_0) = \{\init\}, L(s_2) = \{\pass\} $. 
Note that from $ s_1 $ and $ s_2 $, there are only self-loops.
\end{example}

\begin{definition} 
\label{def:history}
A \emph{path} $\pi$ is a non-empty sequence $s_0 \jact^0 s_1 \jact^1 \dots $ of states and joint actions, 
where $\jact^i \in \ACT^{\AG}$ is the $i^{th}$ joint action, and
$\Delta(s_i,\jact^i, s_{i+1})>0$. 
A \emph{history} $\rho$ is a finite path,
$\rho_s(i)$ denotes the $i^{th}$ state of $\rho$, and $\rho_{\jact}(i)$ denotes the $i^{th}$ joint action of $\rho$. In this case, we may write $\rho_s(i)\TRANS{\rho_{\jact}(i)} \rho_s(i+1)$.
Let $\HIST{s}{\Game}$ denote the set of histories of $\Game$ starting from state $s$. 
We write $\HIST{s}{\Game}^{\leq k}$ to denote the set of histories up to and including $\rho_{s}(k)$.
\end{definition}

In this work, we assume that players have memoryless strategies. Informally, a memoryless strategy for player $i$ prescribes, from each state $s$, the probability of each action $a \in \ACT_i(s)$ being chosen. 

\begin{definition}
\label{def:strategy}
A \emph{(memoryless) strategy} for $i$ is a function from the set of states to a probability distribution over agent's set of actions $\sigma_i: S \FS \dist(\ACT_i)$.
A \emph{strategy profile} is a tuple of \emph{strategies} for a set of agents, it is denoted by $\vec{\sigma}=(\sigma_1, \sigma_2, \dots, \sigma_n)$. For a set (or \textit{coalition}) of agents $J \subseteq \AG$, write $\vec{\sigma}_J$ for $(\sigma_i)_{i \in J}$.
\end{definition}

Note that a strategy profile $\strpElm$ for a game $\Game$ resolves nondeterminism in the game arena. Denote by $\Game_{\strpElm}$ the resulting game arena (essentially, a Markov chain) after applying $\strpElm$. 

\begin{definition}
\label{def:outcome}
A state $s$ and a strategy profile $\vec{\sigma}$ induce a set of histories $s_0 \jact^0 s_1 \jact^1 \dots $, with $\PPP(\jact^i) \cdot \Delta(s_i,\jact^i,s_{i+1})>0$ and $s = s_0$, denoted by $\HIST{s}{\vec{\sigma}}$.
For a history $ \rho \in \HIST{s}{\vec{\sigma}}^{\leq k} $, the probability of $\rho = s_0 \TRANS{\jact^0} s_1 \dots \TRANS{\jact^{k-1}} s_k$, with $s_0 = s$, is given by:
\[
\PPP(\rho) \triangleq \prod^{k-1}_{j=1} \left( \prod^n_{i=1} \left( \sigma_i (s_j) (\jact^j_i) \right)\right)
\]
\end{definition}

For a subset of agents $J \subseteq \AG$ and strategies $\vec{\sigma}_{J}$, we say that a history $\rho$ is \textit{compatible} with $\vec{\sigma}_{J}$ if for every $k \in \NNN$, there exists a joint action $\jact^k$ with $ \strElm_i(\rho_s(k))(\jact^k_i) > 0 $ for each $ i \in J $, such that $ \Delta(s_k,\jact^k,s_{k+1}) > 0 $. We denote by $ \HIST{s}{\vec{\sigma}_{J}} $ the set of histories starting from $ s $ and compatible with $ \strpElm_{J} $ that can be induced by memoryless strategies.

\begin{definition}
\label{def:payoff}
The \emph{payoff} function defined as a map from a set of histories to a real value
$\payoff:\HIST{s}{\Game} \to \RRR^{|\AG|}$. $\payoff^i$ denotes the payoff function of $i\in \AG$ and is defined as:
$$\payoff^i(\rho) \triangleq \left(\sum^{t-1}_{j=0} (r^i_a(\jact^j)+r^i_s(s_j)) \cdot \Delta(s_j,\jact^j,s_{j+1}) \right)$$
where $\rho = s_0 \TRANS{\jact^0} s_1  \TRANS{\jact^1} \dots  \TRANS{\jact^{t-1}}  s_t \in \HIST{s_0}{\Game}$.
\end{definition}

\subsection{PATL}
PATL~\cite{ChenL//:07} is a logic for reasoning about the strategic abilities of (coalitions of) agents in probabilistic systems. In this paper, we consider PATL with \textit{bounded temporal properties}.

\begin{definition}[Syntax of PATL]
\label{def:ATL*-syntax}
The \emph{state formulae} ($\phi$) and 
\emph{path formulae} ($\psi$) are defined as follows:
\begin{GRAMMAR}
 \phi
     &::=&
  a 
     \VERTICAL
  \neg \phi
     \VERTICAL
  \phi \land \phi
     \VERTICAL
    \langle A \rangle \P_{\bowtie p} \lbrack {\psi} \rbrack
      \\
  \psi
     &::=&
  \X \phi 
     \VERTICAL
  (\phi \U_{\leq k} \phi)
\end{GRAMMAR}
where $a \in \ASP$ is an \emph{atomic proposition},
$A \subseteq \AG$ is a set of agents,
${\bowtie} \in \{\le, <, \ge, >\}$,
$p \in \lbrack 0,1 \rbrack$,
$k \in \mathbb{N}$ is a time bound,
and $\langle A \rangle$ is the strategy quantifier.
\end{definition}
The until operator $\U$ allows one to derive the temporal modalities $\F$ (``eventually'') and $\G$ (``always''): $\F_{\le k} \psi  \triangleq  \true~ \U_{\le k}~ \psi$ and $\G_{\le k} \psi  \triangleq  \neg \F_{\le k} (\neg\psi)$. 

\begin{definition} [Semantics of PATL]
\label{def:PATL-semantics}
The satisfaction relation for a CSG $\Game$, state $s \in S$, history $\rho$, atom $a \in \ASP$, and PATL formula is defined as follows:
\begin{itemize}
\setlength{\itemindent}{-1em}

\item $s \models a$ iff $a \in L(s)$.

\item $s \models \neg\phi$ iff $s \not \models \phi$.

\item $s \models \phi \land \phi'$ iff 
  	$s \models \phi$ and  $s \models  \phi'$.

\item $s \models  \langle A \rangle \P_{\bowtie p} [\psi] $ iff
$\exists \strpElm_A . \forall \strpElm_{-A}. (\PPP ( \{ \rho \in  \HIST{s}{(\strpElm_A,\strpElm_{-A})}  \mid \rho \models \psi \} ) \bowtie p)$.

\item $\rho \models \X \phi$ iff $\rho_s(1) \models \phi$.
  
\item $\rho \models \phi \U_{\le k} \phi'$ iff there exists $i \le k$ such that:
  $\rho_s(i) \models \phi'$, and $\rho_s(j) \models \phi$  for all $j < i$.

\end{itemize}
\end{definition}
We say that formula $\varphi$ is {\em true} in a CGS $\Game$ iff $s^0 \models \varphi$.

\section{Formalising Counterfactual Responsibility Attribution} 
\label{sec:cr-notion}

We build on the concept of \textit{Necessary Element of a Sufficient Set} (NESS) test \cite{wright1985causation}, which has been used in many game-theoretic frameworks (e.g., \cite{BrahamVH//:12,baltag2021causal,ParkerGL//:23}), and adapt it to our setting. 

\begin{definition}
\label{def:bcr}
Given a CSG $\Game$, 
we say that agent $i$ bears backward Counterfactual Responsibility (bCR) in a strategy profile $\vec{\sigma}$ for outcome $\varphi$, specified in a PATL path formula, 
written as $i \sim \BCR^{\vec{\sigma}}_{\varphi}$, if the following conditions hold: 
$\exists J \subseteq \AG \setminus \{i\}$ such that 
	$\exists \rho' \in \HIST{s^0}{\strpElm_{J}}. \rho' \not\models \varphi $ and $ 
	\forall \rho'' \in \HIST{s^0}{\strpElm_{ J \cup \{i\} } }. \rho'' \models \varphi$. 
\end{definition}
In words, agent $i$ is responsible for an outcome $\varphi$ under strategy profile $\strpElm$, if it belongs to some coalition $J \cup \{i\}$ whose joint strategies are sufficient to guarantee $\varphi$, while the strategies of $J$ (without $i$) are not. 
That is, counterfactually, changing $i$'s strategy could have prevented $\varphi$.

In the definition above, bCR is a qualitative notion: $i$ either bears responsibility or does not. Our aim is to extend this to a quantitative setting by determining ``how much'' responsibility agent $i$ bears. Intuitively, this corresponds to measuring $i$'s marginal contribution to increasing the likelihood/probability of satisfying $\varphi$. To this end, we allocate responsibility using the Shapley value. A challenge arises when measuring the marginal contribution of a coalition that does not contain all agents (i.e., non-grand coalition) since the strategies of agents outside that coalition are undefined. We address this by assuming that $\varphi$ is an undesirable outcome that agents (retrospectively) would prefer to avoid. Thus, when computing marginal contributions, we assume that agents outside the coalition act so as to minimise the probability of satisfying $\varphi$\footnote{Although we use ``minimise'' when defining the strategies of agents outside the coalition, one could instead use ```maximise'' when $\varphi$ represents a \textit{desirable} outcome. We conjecture that the properties established in the next section would similarly hold.}.

\begin{definition}
\label{def:bcr-degree}
Given a bound\footnote{This bound comes from the PATL path formula $ \varphi $.} $ k \in \NNN$, the degree of $i$ bearing bCR in $\strpElm$ for $\varphi$ is defined as:
$\DDD^i(\strpElm, \varphi) = \sum_{J \subseteq  \AG \setminus \{i\}} \frac{|J|!(|\AG| - |J| - 1)!}{|\AG|!} \cdot (v_{\strpElm, \varphi}(J \cup \{i\}) - v_{\strpElm, \varphi}(J)) $.
For any $A \subseteq \AG$, the function $ v_{\strpElm, \varphi}(A) $ for the $k$-bounded histories is defined as:
\[ v_{\strpElm, \varphi}(A) = \min_{\strpElm_{-A}} \prob(\{\rho \mid \rho \in \HIST{s^0}{{(\strpElm_{A},\strpElm_{-A})}}^{\leq k}  \land \rho \models \varphi \}) \]
\end{definition}

Intuitively, the definition of $ v_{\strpElm, \varphi}(A) $ assumes that the other players are actively trying to \textit{minimise} the occurrence of $ \varphi $. This reflects the \textit{principle of alternative possibilities} and counterfactual reasoning~\cite{frankfurt2018alternate,belnap1992way,widerker2017moral,shi2025responsibility}: \textit{an agent is considered responsible for an outcome $ \varphi $ under a strategy profile $ \strpElm $ if $ \varphi $ occurs under this strategy profile, and the agent had some alternative strategy that could have prevented $ \varphi $.}

\begin{remark}
    We note that whilst $\min$ is used for calculating $v_{\strpElm, \varphi}(A)$, one may also define it using $\max$ when considering $\varphi$ as a ``desirable'' outcome. We conjecture that the properties in \Cref{sec:desireable-props} would also hold.
\end{remark}

\begin{example}
\label{eg:bcr-degree} 
Continuing the running example, suppose given an outcome $\varphi := \X \crash$ (i.e., crash in the next time-step),
and consider the strategy profile $\strpElm$ corresponding to each agent choosing not to brake with probability of $ 1 $. 
In this case, we obtain:
\begin{align*}
v_{\strpElm, \varphi} (\{1,2\}) &= 1, \qquad v_{\strpElm, \varphi} (\{1\}) = 0.6 \\ 
v_{\strpElm, \varphi} (\{2\}) &= 0.2, \qquad v_{\strpElm, \varphi} (\varnothing) = 0.12 
\end{align*}
Therefore: 
\begin{align*}
\DDD^{1}(\strpElm, \varphi) &= \frac{1}{2} \cdot (0.6 - 0.12) + \frac{1}{2} \cdot (1 - 0.2) = 0.64\\  
\DDD^{2}(\strpElm, \varphi) &= \frac{1}{2} \cdot (0.2 - 0.12) + \frac{1}{2} \cdot (1 - 0.6) = 0.24\\
\end{align*}
\end{example}

Note that $ A_1 $ has a higher degree of responsibility than $ A_2 $, intuitively, because $ A_1 $ has a greater chance of being able to stop (i.e., $0.8$, compared to $ 0.4 $ for $ A_2 $). 

\section{Desirable Properties}
\label{sec:desireable-props}
In this section, we demonstrate the desirable properties of our responsibility attribution: \textit{efficiency}, \textit{symmetry}, \textit{null/dummy player}, \textit{additivity}, and \textit{monotonicity}.

To formalise this, we first introduce the notion of \textit{attributable value}, which represents the total amount of responsibility available for distribution among agents. 
The attributable value for a strategy profile is the difference between the probability of $ \varphi $ occurring under the strategy profile and the minimum probability achievable if agents act optimally (together) to prevent $ \varphi $.

To illustrate, consider an autonomous vehicle agent $i$ approaching a T-junction. It can turn \textit{left} or \textit{right}. Turning left guarantees arriving late (probability 1), while turning right carries a 0.2 probability of lateness. Here, lateness is inevitable with at least a 0.2 probability. Suppose the vehicle adopts a strategy profile of turning left with certainty (probability 1). The attributable value is $ 1-0.2=0.8 $, meaning the maximum responsibility assignable to $ i $ is $ 0.8 $.

\begin{definition}
	Given a strategy profile $ \strpElm $, an outcome $ \varphi $, and a bound $ k \in \NNN $, the attributable value of $\varphi$ \wrt $\strpElm$ is defined as: 
    $\aval(\strpElm,\varphi) = v_{\strpElm,\varphi}(\AG)-\min_{\strpElm} \prob(\{\rho \mid \rho \in \HIST{s^0}{{\strpElm}}^{\leq k}  \land \rho \models \varphi \}) $.
\end{definition}

By \Cref{def:bcr-degree}, we have that: 
$$v_{\strpElm,\varphi}(\varnothing) = \min_{\strpElm_{-\varnothing}} \prob(\{\rho \mid \rho \in \HIST{s^0}{{\strpElm}}^{\leq k}  \land \rho \models \varphi \}$$ 
and since $ \AG \setminus \varnothing = \AG$, therefore $ \strpElm_{\AG} = \strpElm $, and we immediately obtain the following lemma.

\begin{lemma}
\label{lem:min}
    For a strategy profile $\strpElm$ and outcome $\varphi$, it holds: 
    $$ \aval(\strpElm,\varphi) = v_{\strpElm,\varphi}(\AG) - v_{\strpElm,\varphi}(\varnothing) $$
\end{lemma}
\begin{proof}
From the definition of attributable value and bCR degree, we have that: 
$$v_{\strpElm,\varphi}(\varnothing) = \min_{\strpElm_{-\varnothing}} \prob(\{\rho \mid \rho \in \HIST{s^0}{{\strpElm}}^{\leq k}  \land \rho \models \varphi \}$$ 
and since $ \AG \setminus \varnothing = \AG$, therefore $ \strpElm_{\AG} = \strpElm $, and we immediately obtain the lemma.
\end{proof}

\begin{example}
\label{eg:attributable}
Consider \Cref{eg:bcr-degree}, the attributable value is: 
\[\aval(\strpElm,\varphi) = v_{\strpElm,\varphi}(\AG) - v_{\strpElm,\varphi}(\varnothing) = 1-0.12=0.88\]
\end{example}

\paragraph{Fairness}
The first three properties relate to \textit{fairness}. \textit{Efficiency} guarantees that the sum of assigned responsibility values equals the attributable value, ensuring no unassigned responsibility remains. Conversely, it also ensures that the distributed responsibility does not exceed the attributable value, preventing any agent from bearing excess responsibility. \textit{Symmetry} ensures that two agents with identical contributions to an outcome receive equal degrees of responsibility. The \textit{dummy/null player} property ensures that agents who do not contribute to an outcome receive zero responsibility.

\begin{proposition}[Efficiency]
    For a strategy profile $\strpElm$ and outcome $\varphi$, it holds that: 
    $\sum_{i \in \AG} \DDD^{i}(\strpElm,\varphi) = \aval(\strpElm,\varphi).$
\end{proposition}
\begin{proof}

\begin{align*}
\sum_{i \in \AG} \DDD^{i} (\strpElm,\varphi) &= \sum_{i \in \AG}  \sum_{J \subseteq  \AG \setminus \{i\}} \frac{|J|!(|\AG| - |J| - 1)!}{|\AG|!} \cdot (v_{\strpElm, \varphi}(J \cup \{i\}) - v_{\strpElm, \varphi}(J)) \\
&= \sum_{J \subseteq \AG}  \sum_{i \in \AG \setminus J} \frac{|J|!(|\AG| - |J| - 1)!}{|\AG|!} \cdot  (v_{\strpElm, \varphi}(J \cup \{i\}) - v_{\strpElm, \varphi}(J))
\end{align*}
Since for each subset $ J $ the term $ \frac{|J|!(|\AG| - |J| - 1)!}{|\AG|!} $ is the same for each player $ i $ not in $ J $, and there are $ |\AG| - |J| $ such players, thus we can simplify the inner sum:
\begin{align*}
\sum_{i \in \AG} \DDD^{i} (\strpElm,\varphi)
&= \sum_{J \subseteq \AG} \frac{|J|!(|\AG| - |J| - 1)!}{|\AG|!} \sum_{i \in \AG \setminus J} 
 (v_{\strpElm, \varphi}(J \cup \{i\}) - v_{\strpElm, \varphi}(J))
\end{align*}
Observe that the expression $ \sum_{i \in \AG \setminus J} (v_{\strpElm, \varphi}(J \cup \{i\}) - v_{\strpElm, \varphi}(J)) $ is a \textit{telescoping series/sum} with respect to all subsets $ J $, thus we can simplify the inner sum further and obtain:
$\sum_{i \in \AG} \DDD^{i} (\strpElm,\varphi)
	= \sum_{J \subseteq \AG} \frac{|J|!(|\AG| - |J| - 1)!}{|\AG|!} \cdot 
    \left( v_{\strpElm, \varphi}(\AG) - v_{\strpElm, \varphi}(\varnothing) \right)$.
Since $ \sum_{J \subseteq \AG} \frac{|J|!(|\AG| - |J| - 1)!}{|\AG|!} = 1 $, we have:
\begin{align*}
\sum_{i \in \AG} \DDD^{i} (\strpElm,\varphi) &= v_{\strpElm, \varphi}(\AG) - v_{\strpElm, \varphi}(\varnothing).
\end{align*}

By \Cref{lem:min}, we obtain:
\begin{align*}
\sum_{i \in \AG} \DDD^{i} (\strpElm,\varphi) &= \aval(\strpElm,\varphi).
\end{align*}
\end{proof}

\begin{example}
\label{eg:efficiency}
Consider again \Cref{eg:bcr-degree} and \ref{eg:attributable}, note that:
\[\sum_{i \in \AG} \DDD^{i} (\strpElm,\varphi) = 0.64+0.24 = \aval(\strpElm,\varphi) = 0.88\]
\end{example}

Symmetry ensures that two agents with identical contributions towards an outcome will receive equal degrees of individual responsibility.

\begin{proposition}[Symmetry]
For two agents $ i,j  \in \AG$, if:
\[\forall J \subseteq \AG \setminus \{i,j\}, v_{\strpElm, \varphi} (J \cup \{i\}) = v_{\strpElm, \varphi} (J \cup \{j\})\]
then:
\[\DDD^{i}(\strpElm,\varphi) = \DDD^{j}(\strpElm,\varphi)\]
\end{proposition}
\begin{proof}
To prove this claim, we first introduce some notations. Denote by $ \Pi(\AG) $ the set of all permutations of $ \AG $. For every permutation $ \pi \in \Pi(\AG) $, define 
\[ P_i(\pi) := \{ j \in \AG : \pi(j) < \pi(i) \}, \]
that is the set of agents ahead of agent $ i $ when they are ordered according to permutation $ \pi $. Using an alternative formulation for the Shapley value, we can express the bCR degree as follows.
\begin{multline}
	\label{eq:shapley-alt}
	\DDD^{i}(\strpElm,\varphi) = \frac{1}{|\AG|!} \sum_{\pi \in \Pi(\AG)} 
    \left( v_{\strpElm,\varphi}(P_i(\pi) \cup \{i\}) - v_{\strpElm,\varphi}(P_i(\pi)) \right)
\end{multline}
Define a function $ f: \Pi(\AG) \to \Pi(\AG) $ that maps each permutation over the set of agents to another permutation as follows. For every permutation $ \pi $, the permutation $ f(\pi) $ is identical to $ \pi $ except that it swaps agent $ i $ with agent $ j $.
Next we show that the following holds.
	\begin{multline}
		\label{eq:sym}
		v_{\strpElm,\varphi}(P_i(\pi) \cup \{i\}) - v_{\strpElm,\varphi}(P_i(\pi)) = v_{\strpElm,\varphi}(P_j(f(\pi)) \cup \{j\}) - v_{\strpElm,\varphi}(P_j(f(\pi)))
	\end{multline}
There are two possible cases:

\textit{Case 1:} agent $ i $ appears \textit{before} agent $ j $ under permutation $ \pi $, i.e., $ j \not\in P_i(\pi) $. Then $ P_i(\pi) = P_j(f(\pi)) $, as such, $ v_{\strpElm,\varphi}(P_i(\pi)) = v_{\strpElm,\varphi}(P_j(f(\pi))) $. Morever, since both $ P_i(\pi) $ and $ P_j(f(\pi)) $ do not contain $ i $ nor $ j $, it holds that: 
$$ v_{\strpElm,\varphi}(P_i(\pi) \cup \{i\}) = v_{\strpElm,\varphi}(P_j(f(\pi)) \cup \{j\}) $$ Therefore, these two equalities imply \Cref{eq:sym}.

\textit{Case 2:} agent $ i $ appears \textit{after} agent $ j $ under permutation $ \pi $, i.e., $ j \in P_i(\pi) $. Then  $ P_i(\pi) \cup \{i\} = P_j(f(\pi)) \cup \{j\} $, as such, $ v_{\strpElm,\varphi}(P_i(\pi) \cup \{i\}) = v_{\strpElm,\varphi}(P_j(f(\pi)) \cup \{j\}) $. Moreover, $ P_i(\pi) \setminus \{j\} = P_j(f(\pi)) \setminus \{i\} $, and since $ i $ and $ j $ are symmetric, it holds that: 
$$ v_{\strpElm,\varphi}((P_i(\pi) \setminus \{j\}) \cup \{j\}) = v_{\strpElm,\varphi}((P_j(f(\pi)) \setminus \{i\}) \cup \{i\}), $$ 
i.e., $ v_{\strpElm,\varphi}(P_i(\pi)) = v_{\strpElm,\varphi}(P_j(f(\pi))) $. Therefore, these two equalities also imply \Cref{eq:sym}.

From \Cref{eq:shapley-alt} and \Cref{eq:sym}, we have
\begin{multline}
	\DDD^{i}(\strpElm,\varphi) = \frac{1}{|\AG|!} \sum_{\pi \in \Pi(\AG)} \left( v_{\strpElm,\varphi}(P_j(f(\pi)) \cup \{j\}) - v_{\strpElm,\varphi}(P_j(f(\pi))) \right)
\end{multline}
Since $ f $ is bijective, thus $ \Pi(\AG) = \{ f(\pi) : \pi \in \Pi(\AG)\} $, and by setting $ \mu = f(\pi) $ we obtain
\begin{multline}
	\DDD^{i}(\strpElm,\varphi) = \frac{1}{|\AG|!} \sum_{\mu \in \Pi(\AG)} 
    \left( v_{\strpElm,\varphi}(P_j(\mu) \cup \{j\}) - v_{\strpElm,\varphi}(P_j(\mu)) \right)
\end{multline}
And, consequently
$\DDD^{i}(\strpElm,\varphi) = \DDD^{j}(\strpElm,\varphi)$. 
\end{proof}

The dummy/null player property guarantees that agents who do not contribute towards an outcome will not receive any responsibility degree values.

\begin{proposition}[Null Player]
If for all $J \subseteq \AG \setminus \{i\}: v_{\strpElm, \varphi} (J \cup \{i\}) =v_{\strpElm, \varphi} (J) $,
then $\DDD^{i}(\strpElm,\varphi) = 0$.
\end{proposition}
\begin{proof}
If for a given $ J \subseteq \AG \setminus \{i\}: v_{\strpElm, \varphi} (J \cup \{i\}) =v_{\strpElm, \varphi} (J) $, then we have $ (v_{\strpElm, \varphi}(J \cup \{i\}) - v_{\strpElm, \varphi}(J)) = 0 $. Since it holds for every $ J \in \AG \setminus \{i\} $, then
\begin{align*}
	\DDD^i(\strpElm, \varphi) &= \sum_{J \subseteq  \AG \setminus \{i\}} \frac{|J|!(|\AG| - |J| - 1)!}{|\AG|!} \cdot 0 = 0.
\end{align*}
\end{proof}

\paragraph{Consistency}
Additivity property states that for two \textit{independent} games, the Shapley value of the combined game is equal to the sum of individual values. 
However, this raises the question of how to combine the games and to what extent they are truly independent\footnote{There have been many attempts to relax the additivity axiom, e.g. \cite{young1985monotonic,hart1989potential,chun1989new,van2002axiomatization}.}, particularly, in our setting. One could define this by considering two separate CSGs, resulting in a ``combined'' responsibility degree that is simply the sum of the two games. However, we find this approach uninteresting and not particularly useful. Instead, we consider additivity with respect to disjunctions of outcomes $ \varphi_1 \lor \varphi_2 $. 
Observe that the function $ \DDD^i(\cdot, \cdot) $ is generally not additive. To see this, consider $ \varphi_1 := \varphi $ and $ \varphi_2 := \neg \varphi $. Clearly, for any $ \strpElm $, $ \DDD^i(\strpElm,\varphi_1 \vee \varphi_2) = 0 $, while it may be the case that $ \DDD^i(\strpElm,\varphi_1) > 0 $ or $ \DDD^i(\strpElm,\varphi_2) > 0 $. However, we can show that \textit{subadditivity} holds, and for a special case of outcomes we obtain (exact) additivity. 
\textit{Monotonicity} intuitively states that for two outcomes $ \varphi_1 $ and $ \varphi_2 $, where $ \varphi_1 $ is ``included'' in $ \varphi_2 $, the degree of responsibility for $ \varphi_2 $ would likely increase.
Additivity and monotonicity relate to consistency in the sense that if we expand the outcome (under certain conditions), then responsibility should increase or at least not decrease.

\begin{proposition}[Sub-additivity]
\label{prop:sub-add}
Given $\strpElm$, for any two outcome $\varphi_1$ and $\varphi_2$, and any $i \in \AG$:
\[\DDD^i(\strpElm, \varphi_1 \lor \varphi_2) \le \DDD^i(\strpElm, \varphi_1) + \DDD^i(\strpElm,\varphi_2)\]
\end{proposition}
\begin{proof}
For any fixed coalition $ A \subseteq \AG $ and any strategy profile $ (\strpElm_A, \strpElm_{-A}) $, we have:
\begin{multline}
\label{eq:sub-additivity}
\ \prob(\{\rho \mid \rho \in \HIST{s^0}{(\vec{\sigma}_A, \strpElm_{-A})}^{\le k} \land \rho \models (\varphi_1 \lor \varphi_2)\} )  \\
\leq \ \big( \prob(\{\rho \mid \rho \in \HIST{s^0}{(\vec{\sigma}_A, \strpElm_{-A})}^{\le k} \land \rho \models \varphi_1 \}) \\
 \ + \prob(\{\rho \mid \rho \in \HIST{s^0}{(\vec{\sigma}_A, \strpElm_{-A})}^{\le k} \land \rho \models  \varphi_2 \}) \big)
\end{multline}
Since the value of the sum at the joint minimiser is no more than the sum of the individual minima, which may occur at different points: 
\[\min_x \lbrack f(x)+g(x) \rbrack \le \min_x f(x) + \min_x g(x)\]
taking the minimum over $ \strpElm_{-A} $ on both sides of \Cref{eq:sub-additivity}, we have:
\[
v_{\strpElm, \varphi_1 \lor \varphi_2}(A) \leq v_{\strpElm, \varphi_1}(A) + v_{\strpElm, \varphi_2}(A)
\quad \text{for all } A \subseteq \AG
\]
By monotonicity of the Shapley value under pointwise function, we obtain the result.
\end{proof}

Intuitively, this holds because the worst-case probabilities for $\varphi_1 \lor \varphi_2$ may overlap, so responsibility for the disjunction can be ``shared'' or even ``cheaper'' than addressing both outcomes individually.

\begin{example}
\label{eg:sub-add}
Continue the running example, suppose given outcomes $\varphi_1=\X \ \crash$ and $\varphi_2=\X \ \pass$, consider $\strpElm$ corresponding to each agent choosing \emph{to brake} with probability of 1, let $\lor$ denote $\varphi_1 \lor \varphi_2$ for short. In this case, we have:
\begin{align*}
&v_{\strpElm,\!\varphi_1}\!(\{1,2\})\!=\!v_{\strpElm,\!\varphi_1}\!(\{1\})\!=\!v_{\strpElm,\!\varphi_1}\!(\{2\})\!=\!v_{\strpElm,\!\varphi_1}\!(\varnothing) = 0.12, \\
&v_{\strpElm, \varphi_2} (\{1,2\}) = 0.88,\quad v_{\strpElm, \varphi_2} (\{1\}) = 0.8, \\
&v_{\strpElm, \varphi_2} (\{2\}) = 0.4,\qquad v_{\strpElm, \varphi_2} (\varnothing) = 0, \\
&v_{\strpElm, \lor} (\{1,2\}) = v_{\strpElm, \lor} (\{1\}) = v_{\strpElm, \lor} (\{2\}) = v_{\strpElm, \lor } (\varnothing) = 0.12 
\end{align*}
Thus: 
\begin{align*}
\DDD^{1}(\strpElm, \varphi_1) &= \DDD^{2}(\strpElm, \varphi_1) = 0, \\
\DDD^{1}(\strpElm, \varphi_2) &= 0.64, \ \DDD^{2}(\strpElm, \varphi_2) = 0.24, \\
\DDD^{1}(\strpElm, \varphi_1 \lor \varphi_2) &= \DDD^{2}(\strpElm, \varphi_1 \lor \varphi_2) = 0
\end{align*}
Clearly, for $i=\{1,2\}:$ 
$\DDD^{i}(\strpElm, \varphi_1 \lor \varphi_2) \le \DDD^{i}(\strpElm, \varphi_1) + \DDD^{i}(\strpElm, \varphi_2)$.
\end{example}

To obtain exact additivity, we define the following conditions.

\begin{definition}
	\label{def:disjoint-avoidable}
	For two outcomes $ \varphi_1 $ and $ \varphi_2 $, we say that they are \textit{disjoint} and \textit{avoidable}, respectively, if
	\begin{itemize}
		\item [i)] for all $ \strpElm $, we have $ \prob(\{\rho \mid \rho \in \HIST{s^0}{{\strpElm}}^{\leq k}  \land \rho \models \varphi_1 \wedge \varphi_2 \}) = 0 $
		\item [ii)] there exists $ \strpElm $, such that $ \prob(\{\rho \mid \rho \in \HIST{s^0}{\strpElm} \rho \models (\varphi_1 \lor \varphi_2)\}) <1 $.
	\end{itemize}
\end{definition}

\begin{proposition}[Exact Additivity]
\label{prop:sup-add}
Let $\varphi_1$ and $\varphi_2$ be two outcome formulas such that they are disjoint and avoidable,
then for any $i \in \AG$:
\[\DDD^i(\strpElm, \varphi_1 \lor \varphi_2) = \DDD^i(\strpElm, \varphi_1) + \DDD^i(\strpElm,\varphi_2)\]
\end{proposition}
\begin{proof}
By the disjointness assumption on $ \varphi_1 $ and $ \varphi_2 $, and the standard property of probability over disjoint events, we have:
\begin{multline}
\label{eq:cond-additivity}
 \ \prob(\{\rho \mid \rho \in \HIST{s^0}{\strpElm_A, \strpElm_{-A}}^{\le k} \land \rho \models (\varphi_1 \lor \varphi_2)\} )  \\
= \ \prob(\{\rho \mid \rho \in \HIST{s^0}{\strpElm_A, \strpElm_{-A}}^{\le k} \land \rho \models \varphi_1 \}) \\
 \ + \prob(\{\rho \mid \rho \in \HIST{s^0}{\strpElm_A, \strpElm_{-A}}^{\le k} \land \rho \models  \varphi_2 \})
\end{multline}
Applying the pessimistic valuation to both sides of \Cref{eq:cond-additivity}, we obtain:
\[
v_{\strpElm, \varphi_1 \lor \varphi_2}(A) \geq v_{\strpElm, \varphi_1}(A) + v_{\strpElm, \varphi_2}(A)
\quad \text{for all } A \subseteq \AG
\]
Since $\varphi_1 \lor \varphi_2$ is avoidable (thus $\DDD^i(\strpElm, \varphi_1 \lor \varphi_2)>0$) by the assumption, by monotonicity of the Shapley value under pointwise function, we have:
\[\DDD^i(\strpElm, \varphi_1 \lor \varphi_2) \ge \DDD^i(\strpElm, \varphi_1) + \DDD^i(\strpElm,\varphi_2)\]
By combining with \Cref{prop:sub-add}, we obtain the proposition.
\end{proof}

Intuitively, when outcomes are disjoint, preventing either one contributes uniquely to the disjunction. This prevents overlap in marginal contributions and can make the total responsibility grow.

\begin{example}
\label{eg:sup-add}
Note that the scenario proposed in \Cref{eg:sub-add} does not satisfy \Cref{prop:sup-add}, as it meets assumption i) but violates assumption ii).
\end{example}

\begin{proposition}[Monotonicity]
\label{prop:mono}
If $\varphi_1 \Rightarrow \varphi_2$ and $\varphi_2$ is avoidable, then for any $i \in \AG$ and a given $\strpElm$: 
\[ \DDD^i(\strpElm, \varphi_1) \le \DDD^i(\strpElm, \varphi_2)\]
\end{proposition}
\begin{proof}
For every coalition $ A \subseteq \AG $, and any strategy profile $ (\strpElm_A, \strpElm_{-A}) $, the implication $ \varphi_1 \Rightarrow \varphi_2 $ ensures:
\begin{align*}
& \prob(\{\rho \mid \rho \in \HIST{s^0}{\strpElm_A, \strpElm_{-A}}^{\le k} \land \rho \models \varphi_1 \}) \\
\leq & \quad
\prob(\{\rho \mid \rho \in \HIST{s^0}{\strpElm_A, \strpElm_{-A}}^{\le k} \land \rho \models \varphi_2 \})
\end{align*}
Therefore, taking the minimum over $ \vec{\sigma}_{-A} $, we get:
\[
v_{\strpElm, \varphi_1}(A) \leq v_{\strpElm, \varphi_2}(A)
\quad \text{for all } A \subseteq \AG
\]
Since $\varphi_2$ is avoidable (thus $\DDD^i(\strpElm, \varphi_2)>0$) by the assumption, and the Shapley value is monotonic under pointwise comparison of value functions, it follows that:
$ \DDD^i(\strpElm, \varphi_1) \le \DDD^i(\strpElm, \varphi_2).$
\end{proof}

Intuitively, responsibility should not decrease when the outcome becomes more inclusive - any trace satisfying $\varphi_1$ also satisfies $\varphi_2$, so agents influencing $\varphi_1$ also influencing $\varphi_2$, potentially to a greater extent.

\begin{example}
\label{eg:monotone}
Continue the running example, consider $\varphi_1 {=} \X \crash$ and $\varphi_2 {=} \F_{\le 2} \ \crash$ and $\strpElm$ corresponding to each agent choosing \emph{not to brake} with probability of 1. Clearly, in this case, $\varphi_1 \Rightarrow \varphi_2$ and $\varphi_2$ is avoidable. We have: 
$ \DDD^1(\strpElm, \varphi_1) = \DDD^1(\strpElm, \varphi_2) = 0.64 $,
$\DDD^2(\strpElm, \varphi_1) = \DDD^2(\strpElm, \varphi_2) = 0.24$.
We have ``$=$'' specifically here as there is only a self-loop from $s_1$.
\end{example}

\section{Logic with CR attributions: \rpatl}

We introduce \rpatl, a variant of PATL that incorporates quantified reward and responsibility attribution formulae as defined in \Cref{def:bcr-degree}.

\begin{definition}
\label{def:syntax}                                                         
The \emph{syntax} of \rpatl is made up of \emph{state formulae} and \emph{history formulae} represented by $\phi$ and $\psi$, respectively.
\begin{GRAMMAR}
 \phi
     &::=&
   a 
      \VERTICAL
   \neg \phi
      \VERTICAL
   \phi \land \phi
      \VERTICAL
    \langle A \rangle \P_{\bowtie p} \lbrack {\psi} \rbrack
     \VERTICAL
    \langle A \rangle \R_{\bowtie q} \lbrack {\psi} \rbrack
     \VERTICAL 
     \\
     &&
    \langle A \rangle \D_{\bowtie d} \lbrack {\BCR_{i}(\vec{\sigma},\psi)} \rbrack
      \\
  \psi
     &::=&
  \X \phi 
     \VERTICAL
  \phi  \U_{\le k} \phi 
\end{GRAMMAR} 
$q, d \in \RRR$ are reward and responsibility degree bounds, respectively. 
\end{definition}
The formula $\langle A \rangle \R_{\bowtie q} \ \lbrack \F_{\le k} \psi \rbrack$ expresses that the coalition $A$ has a strategy such that the expected rewards of satisfying path formula $\psi$ is $\bowtie q$ when the strategy is followed.
The formula $\langle A \rangle \D_{\bowtie d} \ \lbrack \BCR_i (\vec{\sigma}, \psi) \rbrack$ expresses that, under strategy profile $\vec{\sigma}$, the backward counterfactual responsibility of $i$ in ensuring $\psi$ within coalition $A$ is quantified by $\bowtie d$.

\begin{definition}
\label{def:semantics}
Given a model $\Game$, semantics for \rpatl are interpreted as follows.

For a state $s \in S$ of $\MMM$, the \emph{satisfaction relation} $s \models \phi$ for state formula denotes ``$s$ satisfies $\phi$'':
\begin{itemize}

\item $s \models a$ iff $a \in L(s)$.

\item $s \models \neg\phi$ iff $s \not \models \phi$.

\item $s \models \phi \land \phi'$ iff 
  	$s \models \phi$ and  $s \models  \phi'$.

\item $s \models  \langle A \rangle \P_{\bowtie p} [\psi] $ iff \\
$\exists \strpElm_A . \forall \strpElm_{-A}. (\PPP ( \{ \rho \in  \HIST{s}{(\strpElm_A,\strpElm_{-A})}  \mid \rho \models \psi \} ) \bowtie p)$.

\item $s \models  \langle A \rangle \R_{\bowtie q} \lbrack \psi \rbrack$ iff there exists a strategy profile of coalition $A$ such that the expected accumulated reward of paths from state $s$ which are consistent with the strategy and satisfy $\psi$ is $\bowtie q$, \ie 
$\exists \vec{\sigma}_A. \EEE^{\vec{\sigma}_A}_{\Game} \big(s, \rew(r, \psi) \big) \bowtie q$,
where: 
$$
\rew(r, \psi) (\rho) = 
\left \lbrace 
\begin{array}{ll}
\sum^{k}_{i=0} (r_a(\rho_a(i)) + r_s(\rho_s(i))) & \text{if } \rho \models \psi;\\
0 & \text{otherwise} 
\end{array}
\right.
$$
where $k$ denotes the time bound of $\psi$. 

\item $s \models \langle A \rangle \D_{\bowtie d} \lbrack \BCR_i(\strpElm, \psi) \rbrack $ 
iff \ $\DDD^i (\strpElm_{A}, \psi) \bowtie d$, where:
$$\DDD^{i} (\vec{\sigma}_{A}, \psi) = \sum_{J \subseteq A \setminus \{i\}} \tfrac{|J|! (|A| - |J| - 1)}{|A|!} \cdot (v_{\strpElm_{A}, \psi} (J \cup \{i\}) - v_{\strpElm_{A}, \psi}(\{i\})).$$
    
\end{itemize}
    
For a history $\rho \in \HIST{s_0}{\MMM}$, the \emph{satisfaction relation}
$\rho \models \psi$ for a path formula $\psi$ denotes that ``$\rho$ satisfies $\psi$'':
\begin{itemize} 
    \item $\rho \models \X \phi$ iff $\rho_s(1) \models \phi$.
    \item $\rho \models \phi \U_{\le k} \phi'$ iff there exists $i \le k$ such that:
  $\rho_s(i) \models \phi'$, and $\rho_s(j) \models \phi$  for all $j < i$.
\end{itemize}

\end{definition}

\begin{example}
\label{eg:logic}
Consider our running example with the outcome that ``a crash takes place within 2 time steps'', which can be expressed by: 
$\psi = \langle \{1, 2\} \rangle \F_{\le 2} \ \crash$.
Suppose $1$ and $2$ both choose to brake under $\strpElm$, thus:
$v_{\strpElm,\psi}(\{1,2\}) = v_{\strpElm,\psi}(\{1\}) = v_{\strpElm,\psi}(\{2\}) = v_{\strpElm,\psi}(\varnothing) = 0.12$, and thus $\DDD^1(\strpElm, \psi) = \DDD^2(\strpElm, \psi) = 0$.
So: $\D_{\le 0} \lbrack \BCR_i(\vec{\sigma}, \psi) \rbrack$ is $\true$. 
\end{example}

\section{Model Checking \rpatl}
The model checking of \rpatl works in a similar way to rPATL proposed in \cite{KwiatkowskaGPS//:19}. The main difference is that we have the new $\D_{\bowtie d} \lbrack {\BCR_{i}(\vec{\sigma},\psi)} \rbrack$ operator, which corresponds to computing the function $\DDD^{i} (\vec{\sigma}_{A}, \psi)$ in \Cref{def:semantics}, which is presented in Algorithm \ref{alg:dbcr}. Computing such function can be done in polynomial space.

\begin{algorithm}[tbh]
\caption{degree-BCR: Calculate $\DDD^i(\strpElm_A, \psi)$}
\begin{algorithmic}[1]
\State \textbf{Input:} $\MMM, i, \strpElm, A, \psi, k$
\State \textbf{Output:} $\DDD^i(\strpElm_A, \psi)$

\State $d \gets 0$

    \For{each $J \subseteq A \setminus \{i\}$}
    
    \State 
    	$ v_1 \gets \min_{\strpElm_{-(J \cup \{i\})}} \prob(\{\rho \mid \rho \in \HIST{s^0}{{(\strpElm_{(J \cup \{i\})},\strpElm_{-(J \cup \{i\})})}}^{\leq k}  \land \rho \models \psi \}) $
    	
     \State 
    $ v_2 \gets \min_{\strpElm_{-J}} \prob(\{\rho \mid \rho \in \HIST{s^0}{{(\strpElm_{J},\strpElm_{-J})}}^{\leq k} \land \rho \models \psi \}) $

	\State $d \gets d + \tfrac{|J|! (|\AG| - |J| - 1)}{|\AG|!} \cdot (v_1-v_2)$
	
    \EndFor

\State \textbf{Return} $d$
\end{algorithmic}
\label{alg:dbcr}
\end{algorithm}

\begin{lemma}
\label{lem:alg-1}
    Given a CSG $\Game$, a coalition $A$, strategy profile $\strpElm$, and a path formula $\psi$, computing $\DDD^{i} (\vec{\sigma}_{A}, \psi)$ can be done in polynomial space.
\end{lemma}

\begin{proof}
To show that it can be computed in polynomial space, we propose \Cref{alg:dbcr}. We show that the algorithm runs in PSPACE. 

We first examine lines 5 and 6. Observe that computing:
$$ \min_{\strpElm_{-J}} \prob(\{\rho \mid \rho \in \HIST{s^0}{{(\strpElm_{J},\strpElm_{-J})}}^{\leq k}  \land \rho \models_{\MMM} \psi \}) $$ 
corresponds to model checking a \textit{probabilistic} CTL (PCTL) formula over a \textit{Markov Decision Process} (MDP) \cite{BaierK//:08}. When we apply the (partial) strategy profile $ \strpElm_{J} $, we obtain an MDP $ \Game_{\strpElm_{J}} $. 

We then compute the quantitative value of the PCTL formula $ \P_{\min=?} \psi $, which represents the minimum probability value of $ \psi $ in the MDP $ \Game_{\strpElm_{J}} $. This can be solved in polynomial time using linear programming and value iteration method \cite{BaierK//:08}. Therefore, lines 5 and 6 are solvable in polynomial time. 

Furthermore, the values $ |J|! $ and $ |\AG|! $ require $ O(n \log n) $ bits. Since for each subset $ J $, we only need to keep the current intermediate result in memory, and each intermediate result is of polynomial size, the PSPACE membership follows immediately.
\end{proof}

\begin{theorem}
  \label{thm:complexity}
Model checking \rpatl\ is in PSPACE.	
\end{theorem}

\begin{proof} 
	The proof consists of two parts. 	
    
	First, we show that model checking \rpatl\ formula without $ \D $ operators is in PSPACE. This follows from the fact that the fragment of \rpatl\ without $ \D $ operators corresponds to the logic rPATL proposed in \cite{KwiatkowskaGPS//:19}, whose model checking over CSGs is shown to be in PSPACE. 
    
    Second, by \Cref{lem:alg-1}, model checking \rpatl\ formulae of the form $\langle A \rangle \D_{\bowtie d} \lbrack {\BCR_{i}(\vec{\sigma},\psi)} \rbrack$ can also be done in PSPACE. Therefore, the \rpatl\ model checking problem can be solved in PSPACE.
\end{proof}

\section{CR-aware Strategic Decision-Making}
\label{sec:resp-aware}
This section addresses how agents should (or, would) choose strategies when they are aware of counterfactual responsibility attributions. We introduce a utility function (defined shortly) that incorporates both the agent's payoff and its degree of responsibility. To identify stable outcomes, we use \textit{Nash equilibrium} (NE) assuming agents act rationally to maximise their utility.

First, we introduce a parametric model. Observe that the set of all memoryless strategies for player $i$ from state $s$ can be encoded by a set of variables $V^i_s = \{ x_a : a \in \ACT_i(s) \}$. Intuitively, the value of $x_a \in V^i_s$ corresponds to the probability of action $a$ being chosen by player $i$ in state $s$.
Let $ V^i = \bigcup_{s \in S} V^i_s $.
A memoryless (mixed) strategy for $ i $ in $ \Game $ thus corresponds to an \textit{evaluation} of such a set of variables represented by a function $ C^i : V^i \to \mathbb{R} $. 

\begin{definition}
\label{def:psmas}
Given $\Game = (\AG, S, s^0, (\ACT_i)_{i \in \AG}, \delta, \ASP, L )$, we construct the corresponding \emph{parametric stochastic multi-agent system} (\psmas) as a tuple 
$\MMM{=}(\AG, S, s^0, \ACT, V, \Delta, \ASP, L)$, where: 
$V^i = \{x_{i,1}, \dots, x_{i,k}\} \subseteq V$ is a finite set of variables (\emph{parameters}) over $\RRR^k$ for each agent $i$; 
$\Delta: S \times \ACT^{\AG} \times S \to \FFF_{x} $ is the \emph{probabilistic transition function}, where $\FFF_x$ is the set of polynomials over $V$ with rational coefficients which can be viewed as a parametric transition probability matrix that respects the distribution in $ \delta $.
\end{definition}
Observe that the resulting \psmas $\MMM$ is polynomial in size wrt the game $\Game$. Next, we introduce the notion of \textit{admissible} evaluations in a given \psmas. An evaluation $ C^i $ is admissible if: 
\begin{enumerate}
	\item [(1)] $ \Delta_{C^i}(s,\jact,s) \in [0,1] $ for all $ s,s' \in S $ and $ \jact \in \ACT^{\AG} $, i.e., each joint action has a probability between 0 and 1;
	\item [(2)] $ C^i(x) \in [0,1] $ for all $ i \in \AG $ and $ x \in V^i $, i.e., for each player, each action probability is between 0 and 1;
	\item [(3)] $ \sum_{x \in V^i_s} C^i(x) = 1 $ for all $ i \in \AG $ and $ s \in S $, i.e., For each player and each state, the total of action probability values should equal to 1.
\end{enumerate}
Henceforth, we assume that the model takes the form of a \psmas. When evaluating solutions, we further assume that such solutions are admissible. 

\subsection{Responsibility-aware Utility Function}
\begin{definition}[Responsibility-aware Utility Function]
\label{def:utility}
Given $\MMM$ and an outcome $\varphi$ and strategy profile $\strpElm$, for an agent $i$, the utility function is defined as:
	\begin{equation} 
	\label{eq:utility-b}
	u^{\varphi}_i(\strpElm) \triangleq \sum_{\rho \in \HIST{s^0}{\strpElm}}\payoff^i(\rho) - \lambda \cdot \DDD^i (\strpElm, \varphi) 
	\end{equation}
    where $\payoff^i(\rho)$ and $\DDD^i (\strpElm, \varphi)$ are the payoff function and bCR degree function defined in \Cref{def:payoff} and \ref{def:bcr-degree} respectively, $\lambda \in \mathbb{R}$ is a coefficient.
\end{definition}

\begin{definition}[\textsc{NE-Computation}]
    Given a \psmas\ $\MMM$, a \rpatl\ path formula $\varphi$, and utility functions $(u^{\varphi}_i(\strpElm))_{i \in \AG}$, the \textsc{NE-Computation} problem is to compute a NE strategy profile $\strpElm$ that satisfies $\varphi$. 
\end{definition}

The high level process to address \textsc{NE-Computation} consists of two main steps:  
(i) \textit{Satisfaction filtering:} identify the set of strategy profiles that satisfy $\varphi$. This is achieved through parametric model checking on $\MMM$, which yields rational valuation functions $V^i$ for each agent $i \in \AG$. These functions capture the probabilities of selecting actions $a \in \ACT$ in each state $s$ under mixed strategies.  
(ii) \textit{Equilibrium filtering:} refine this set by isolating the strategies that constitute an NE. This involves formulating a system of equations characterising NE conditions and solving them within the strategy set obtained in (i).

\begin{proposition}[Quasi-Concavity of Utility]
\label{prop:quasi}
The utility function $u^{\varphi}_i$ is quasi-concave if the following conditions hold:
\begin{itemize}
\item [1)] linearity of reward: the expected reward is linear in $\strpElm_i$;
\item [2)] linearity of responsibility contribution: the responsibility-valuation function $v_{\strpElm,\varphi}(A)$ is linear in $\strpElm_i$ for all coalition $A \subseteq \AG$;
\item [3)] fixed strategies of others: the strategy profile $\strpElm_{-i}$ of other agents is fixed.
\end{itemize}
\end{proposition}
\begin{proof}
Under condition 1) and 2), both the expected reward and the BCR degree term are linear functions of $\strpElm_i$. 
A linear function is trivially quasi-concave (and also quasi-convex).
Since $\strpElm_{-i}$ is fixed, we are examining a function $u_i(\strpElm_i)$ over a convex domain (probability simplex), which is a linear combination of linear terms. 
Hence $u_i$ is quasi-concave in $\strpElm_i$ under these conditions.
\end{proof}

In our setting, quasi-concavity is evaluated with respect to an individual agent’s mixed strategy. To ensure this is well-defined, we fix the strategies of all other agents. This is a standard approach in both game-theoretic best-response analysis and parametric model checking \cite{ShohamL//:09,BassetKTW//:15}.
When both reward and responsibility allocation are linear in the agent's strategy, the resulting utility function is quasi-concave, enabling NE computation via convex optimisation.

In our PSMAS, strategies are symbolically encoded using parameters that represent the probabilities of actions. When the system and property satisfaction probability can be expressed as a linear function of these parameters, both the reward and responsibility expressions remain linear. Consequently, the symbolic utility is linear. We demonstrate that for outcomes expressible as either reachability or safety objectives, the quasi-concavity property is obtained.

\begin{proposition}
    In a game where the formula $\varphi$ is a reachability or safety formula, the quasi-concavity property (\Cref{prop:quasi}) holds.
\end{proposition}

\begin{proof}
    For the responsibility component, note that when $\varphi$ is either \textit{reachability} or \textit{safety}, pure memoryless strategies are sufficient to satisfy it~\cite{chatterjee2012survey}. Thus, the probability of satisfying $\varphi$ can be expressed as a linear polynomial over strategy parameters $x_i$. Consequently, $v_{\strpElm,\varphi}(A)$ is linear in $\strpElm_i$, and thus $v_{\strpElm,\varphi}(A \cup \{i\}) - v_{\strpElm,\varphi}(A)$ is also linear in $\strpElm_i$. Therefore, $\DDD^i(\strpElm, \varphi)$ is linear.
    
    For the reward component, since pure memoryless strategies are sufficient, the reward can be represented by a linear function. This linearity condition also trivially extends to strategy profiles.
\end{proof}

In general, $u^{\varphi}_i$ is quasi-concave in $\strpElm_i$ when the agent's preference for mixing strategies leads to at least as good an outcome as playing pure strategies, or formally, when the level sets $\strpElm_i \mid u^{\varphi}_i(\strpElm_i, \strpElm_{-i}) \ge c$ are convex for every $c \in \RRR$ \cite{OsborneR//:94}. This result enables the use of convex optimisation (best response over pure strategies) to compute Nash equilibria when the assumptions hold. 

\paragraph{Computing bCR-aware utility.}
The computation of the \textit{responsibility-aware utility valuation function}, as given in \Cref{eq:utility-b}, for outcome $\psi$ 
involves: 
\begin{itemize}
\item [i)] calculating the expected rewards: 
$\EEE (\payoff^i( \{\rho \in \HIST{s}{\strpElm_A} \mid \rho \models_{\MMM} \psi\}) )$, and
\item [ii)] calculating the bCR degree: $\DDD^{i} (\strpElm_{A}, \psi)$.
\end{itemize}
Finally, computing the best response set $\BR_i(\strpElm_{\AG \setminus i})$ corresponds to finding the set of strategies $\sigma_i$ that maximise the utility as the NE.

\subsection{NE for Responsibility-aware Strategy}
\begin{definition}[NE for bCR-aware Strategy]
\label{def:br}
Given $\MMM$, for each agent $i$, and $\strpElm_{\AG\setminus i}$,  a strategy $\sigma_i$ is a \emph{best response} \wrt the utility function defined in (\ref{eq:utility-b}),  if it is the set:
\begin{align*} 
\BR_i(\strpElm_{\AG\setminus i}) \triangleq& \{ \sigma_i \mid \max_{\sigma_i}  (u^{\varphi}_i(\sigma_i, \strpElm_{\AG\setminus i}) ) \}
\end{align*}
A strategy profile $\strpElm$ is a mixed NE if it is a best response for every agent, meaning that $\strpElm \in \BR_i(\strpElm_{\AG \setminus i})$ holds for all $i \in \AG$ in $\MMM$.
\end{definition}

\begin{example}
\label{eg:utility}
Revisiting our running example, consider $ \psi = \X \ \crash$, $A_1$ and $A_2$ choose ``braking'' with probability $x_1$ and `$x_2$ respectively under $\strpElm$. Suppose that the payoffs for $A_1$ and $A_2$ ``braking'' yields 3 and 1 respectively, while ``not-braking'' results in a payoff of 2 and 3 respectively, 
the payoffs of reaching $s_1$ (crash) for $A_1$ and $A_2$ is $-3$ and $-1$ respectively.
The expected payoff can be computed as:
\[\sum_{\rho \in \HIST{s^0}{\strpElm}}\payoff^{1}(\rho) = x_1+0.4x_2-1.6x_1x_2-1\]
\[\sum_{\rho \in \HIST{s^0}{\strpElm}}\payoff^{2}(\rho) = -1.6x_1-2x_2+1.6x_1x_2+2\]
\[v_{\strpElm,\psi}(\{1,2\}) = v_{\strpElm,\psi}(\{2\}) = 0.2 x_1 (1-x_2)\]
\[v_{\strpElm,\psi}(\{1\}) = v_{\strpElm,\psi}(\varnothing) = 0.12 x_1 x_2\]
\[\DDD^{1} (\strpElm,\psi) = 0,\ \DDD^{2} (\strpElm,\psi) = 0.2 x_1 - 0.32 x_1x_2 \]
Consider an example form of the polynomial function: 
\[\sum_{\rho \in \HIST{s^0}{\strpElm}}\payoff^{i}(\rho)  - 10 \cdot \DDD^{i} (\strpElm,\psi)\] 
we obtain the utility function as:
\begin{align*}
u^{\psi}_{1}(\strpElm) &= x_1+0.4x_2-1.6x_1x_2-1\\ 
u^{\psi}_2(\strpElm) &= -3.6x_1-2x_2+4.8x_1x_2+2.
\end{align*}
\end{example}

\paragraph{Parametric best response expressions.}
When performing parametric model checking against \rpatl \ formulas to evaluate best responses, we obtain parametric expressions that encode optimal strategies. These expressions include parameters representing the probabilities of different actions for each agent. By varying these parameters, we can identify best-response strategies within the strategy profile space, \ie action probabilities that maximise each agent's utility. This approach enables a systematic exploration of how changes in action probabilities affect both the overall performance and allocation of responsibility within the strategy profile.

\paragraph{Finding the stable responsibility-aware strategy}
If the utility function $ u^{\varphi}_i $ is quasi-concave (more specifically linear and convex) with respect to agent $i$'s mixed strategies in a given model $\MMM$, then a mixed strategy profile $ \strpElm $ is a NE if and only if every pure strategy played with positive probability in $ \sigma_i $ is a best response to the mixed strategies of the other agents \cite{OsborneR//:94}. 
This implies that, in a mixed NE, an agent randomises only among actions that yield the same expected utility, \ie every action in the support of an agent's equilibrium mixed strategy must lead to an identical valuation:  
$u^{\varphi}_i(\strpElm_{-i}, \sigma_i) = u^{\varphi}_i(\strpElm_{-i}, \sigma_{i\prime})$.
For any two actions $ a, a' \in \ACT $ corresponding to $ \sigma_i $ and $ \sigma_{i\prime} $, provided they have nonzero probabilities.
This property forms the basis for expressing the NE condition:
\begin{equation}\label{eq:ne-poly}
\left\lbrace
\begin{array}{l}
u^{\varphi}_i(\strpElm_{-i}, \sigma_i)=u^{\varphi}_i(\strpElm_{-i}, \sigma_{i\prime}) \quad \forall \sigma_i,\sigma_{i\prime}\\ 
\sum^m_{j=1} \sigma_{ij}=1 \quad \forall i \\
0 \le \sigma_{ij} \le 1 \quad \forall i, j
\end{array}
\right.
\end{equation}
Solving these equations allows us to identify the NE and determine the optimal strategy profile for the given objective. This leads to the following theorem, whose proof follows from a reduction to the existential theory of real numbers~\cite{Canny//:88}.

\begin{theorem}
\label{theo:ne-complexity}
Computing NE strategy satisfying \rpatl \ formula $\varphi$ can be done in PSPACE.
\end{theorem}

\begin{example}
\label{eg:ne}
Revisiting \Cref{eg:utility}, by solving the NE condition equations, we have: $x_1=0.42$, $x_2=0.625$.
\end{example}
\section{Conclusion and Future Work}
\label{sec:concl}
  
We introduce a formal definition of backward (counterfactual) responsibility allocation and provide a framework to model and reason about it. Our definition of responsibility allocation satisfies desirable properties such as fairness and consistency. Utilising this definition, we propose \rpatl\ as a language to strategically reason about responsibility allocation, reward, and temporal objectives in probabilistic multi-agent systems. We demonstrate that the model checking problem is in PSPACE, thus no more complex than rPATL as proposed in \cite{KwiatkowskaGPS//:19}. Furthermore, we present an approach to calculate NE strategy profiles, where players' utility functions consider both reward and responsibility allocation. These strategy profiles correspond to stable behaviour in the game and offer a method to determine how agents should or would choose their strategies. This corresponds to game theory's \textit{normative} (how agents should/ought to act) and \textit{descriptive} (how agents would act) interpretations \cite{wooldridge2012does}. Both interpretations are valuable for the analysis, verification, and design of responsibility-aware MASs.

As discussed in \cite{kwiatkowska2018automated}, bounded temporal properties may require memory. Therefore, a potential future direction is to generalise our approach to \textit{finite-memory} agents. However, this would necessitate a different approach, as using a parametric model would no longer suffice. Another intriguing future direction is to incorporate responsibility allocation into \textit{normative systems}~\cite{aagotnes2007normative}, where norms can be introduced to achieve desirable outcomes from the designer's perspective~\cite{perelli2019enforcing}. Additionally, when integrating learning components into the systems, such as in a \textit{multi-agent reinforcement learning} setting, it is worth exploring how to incorporate responsibility allocation with rewards. This could involve designing \textit{reward structures} or \textit{reward machines}~\cite{icarte2022reward} to achieve desirable outcomes.

\bibliographystyle{splncs03} 
\bibliography{BIB-resp}

\end{document}